\documentclass[aps,pra,twocolumn,twoside,showpacs,superscriptaddress,a4paper]{revtex4-1}

\usepackage[verbose=true]{microtype}
\usepackage{amsmath}
\usepackage{amssymb}
\usepackage{mathdots}
\usepackage{mathrsfs}
\usepackage{paralist}
\usepackage{ragged2e}
\usepackage{graphicx}
\usepackage{color}
\usepackage{varioref}
\usepackage[linktocpage,hypertexnames=false,pdfpagelabels]{hyperref}
\usepackage[amsmath,thmmarks,hyperref]{ntheorem}
\usepackage{cleveref}

\def\clap#1{\hbox to 0pt{\hss#1\hss}}

\newtheorem{theorem}{Theorem}

\newtheorem{corollary}[theorem]{Corollary}

\newtheorem{definition}[theorem]{Definition}

\newtheorem{remark}[theorem]{Remark}

\theoremstyle{break}

\theoremstyle{break}
\theorembodyfont{\normalfont}

\def\squareforqed{\hbox{\rlap{$\sqcap$}$\sqcup$}}
\def\qed{\ifmmode\squareforqed\else{\unskip\nobreak\hfil
\penalty50\hskip1em\null\nobreak\hfil\squareforqed
\parfillskip=0pt\finalhyphendemerits=0\endgraf}\fi}
\def\endenv{\ifmmode\;\else{\unskip\nobreak\hfil
\penalty50\hskip1em\null\nobreak\hfil\;
\parfillskip=0pt\finalhyphendemerits=0\endgraf}\fi}
\newenvironment{proof}{\noindent \textbf{{Proof~} }}{\qed}
\mathchardef\ordinarycolon\mathcode`\:
\mathcode`\:=\string"8000
\def\vcentcolon{\mathrel{\mathop\ordinarycolon}}
\begingroup \catcode`\:=\active
  \lowercase{\endgroup
  \let :\vcentcolon
  }

\newcommand{\nc}{\newcommand}
\nc{\rnc}{\renewcommand}
\nc{\beq}{\begin{equation}}
\nc{\eeq}{{\end{equation}}}
\nc{\beqa}{\begin{eqnarray}}
\nc{\eeqa}{\end{eqnarray}}
\nc{\lbar}[1]{\overline{#1}}
\nc{\bra}[1]{\langle#1|}
\nc{\ket}[1]{|#1\rangle}
\nc{\ketbra}[2]{|#1\rangle\!\langle#2|}
\nc{\braket}[2]{\langle#1|#2\rangle}
\nc{\proj}[1]{| #1\rangle\!\langle #1 |}
\nc{\avg}[1]{\langle#1\rangle}
\nc{\Rank}{\operatorname{Rank}}
\nc{\smfrac}[2]{\mbox{$\frac{#1}{#2}$}}
\nc{\tr}{\operatorname{Tr}}
\nc{\ox}{\otimes}
\nc{\dg}{\dagger}
\nc{\dn}{\downarrow}
\nc{\cA}{{\cal A}}
\nc{\cB}{{\cal B}}
\nc{\cC}{{\cal C}}
\nc{\cD}{{\cal D}}
\nc{\cE}{{\cal E}}
\nc{\cF}{{\cal F}}
\nc{\cG}{{\cal G}}
\nc{\cH}{{\cal H}}
\nc{\cI}{{\cal I}}
\nc{\cJ}{{\cal J}}
\nc{\cK}{{\cal K}}
\nc{\cL}{{\cal L}}
\nc{\cM}{{\cal M}}
\nc{\cN}{{\cal N}}
\nc{\cO}{{\cal O}}
\nc{\cP}{{\cal P}}
\nc{\cR}{{\cal R}}
\nc{\cS}{{\cal S}}
\nc{\cT}{{\cal T}}
\nc{\cX}{{\cal X}}
\nc{\cZ}{{\cal Z}}
\nc{\csupp}{{\operatorname{csupp}}}
\nc{\qsupp}{{\operatorname{qsupp}}}
\nc{\var}{{\operatorname{var}}}
\nc{\rar}{\rightarrow}
\nc{\lrar}{\longrightarrow}
\nc{\polylog}{{\operatorname{polylog}}}
\nc{\1}{{\openone}}
\nc{\wt}{{\operatorname{wt}}}
\nc{\av}[1]{{\left\langle {#1} \right\rangle}}

\def\e{\epsilon}

\nc{\RR}{{{\mathbb R}}}
\nc{\CC}{{{\mathbb C}}}
\nc{\FF}{{{\mathbb F}}}
\nc{\NN}{{{\mathbb N}}}
\nc{\ZZ}{{{\mathbb Z}}}
\nc{\PP}{{{\mathbb P}}}
\nc{\QQ}{{{\mathbb Q}}}
\nc{\UU}{{{\mathbb U}}}
\nc{\EE}{{{\mathbb E}}}
\nc{\id}{{\operatorname{id}}}

\nc{\CHSH}{{\operatorname{CHSH}}}

\nc{\be}{\begin{equation}}
\nc{\ee}{{\end{equation}}}
\nc{\bea}{\begin{eqnarray}}
\nc{\eea}{\end{eqnarray}}
\nc{\<}{\langle}
\rnc{\>}{\rangle}
\nc{\Hom}[2]{\mbox{Hom}(\CC^{#1},\CC^{#2})}
\nc{\rU}{\mbox{U}}

\nc{\ob}[1]{#1}

\nc{\SEP}{{\text{SEP}}}
\nc{\sep}{{\text{sep}}}
\nc{\LOCC}{{\text{LOCC}}}
\nc{\PPT}{{\text{PPT}}}
\nc{\EXT}{{\text{EXT}}}
\nc{\ALL}{{\text{ALL}}}
\nc{\Sym}{{\operatorname{Sym}}}

\nc{\vertleq}{{\rotatebox[origin=c]{90}{$\leq$}}}

\crefname{equation}{Eq.}{Eqs.}
\crefrangelabelformat{equation}{\textup{(#3#1#4)}--\textup{(#5#2#6)}}
\crefname{section}{Section}{Sections}
\crefname{subsection}{Section}{Sections}
\crefname{subsubsection}{Section}{Sections}
\crefname{subsubsubsection}{Section}{Sections}
\crefname{appendix}{Appendix}{Appendices}
\crefname{figure}{Fig.}{Figs.}
\crefname{table}{Table}{Tables}
\crefname{enumi}{Item}{Items}
\crefname{enumii}{Item}{Items}
\crefname{enumiii}{Item}{Items}
\crefname{enumiv}{Item}{Items}
\crefname{enumv}{Item}{Items}
\crefname{theorem}{Theorem}{Theorems}
\crefname{lemma}{Lemma}{Lemmas}
\crefname{corollary}{Corollary}{Corollaries}
\crefname{proposition}{Proposition}{Propositions}
\crefname{definition}{Definition}{Definitions}
\crefname{result}{Result}{Results}
\crefname{equations}{Eqs.}{Eqs.}
\crefname{equations}{Equations}{Equations}
\creflabelformat{equations}{\textup{(#2#1#3)}}
\crefrangelabelformat{equations}{\textup{(#3#1#4)}--\textup{(#5#2#6)}}

\DeclareSymbolFont{txfontsC}{U}{txsyc}{m}{n}
\SetSymbolFont{txfontsC}{bold}{U}{txsyc}{bx}{n}
\DeclareFontSubstitution{U}{txsyc}{m}{n}

\DeclareMathAlphabet{\mathitbf}{OML}{cmm}{b}{it}

\DeclareMathSymbol{\coloneqq}{\mathrel}{txfontsC}{66}
\DeclareMathSymbol{\eqqcolon}{\mathrel}{txfontsC}{67}

\begin{document}

\title{Non-Additivity of the Entanglement of Purification \protect\\ (Beyond Reasonable Doubt)}

\author{Jianxin Chen}
\email{chenkenshin@gmail.com}
\affiliation{Department of Mathematics \& Statistics, University of Guelph, Guelph, Ontario, Canada}
\affiliation{Institute for Quantum Computing, University of Waterloo, Waterloo, Ontario, Canada}

\author{Andreas Winter}
\email{a.j.winter@bris.ac.uk}
\affiliation{Department of Mathematics, University of Bristol, Bristol BS8 1TW, U.K.}
\affiliation{Centre for Quantum Technologies, National University of Singapore, 
2 Science Drive 3, Singapore 117542, Singapore}

\date{5 June 2012}

\begin{abstract}
We demonstrate the convexity of the difference between the
regularized entanglement of purification and the entropy, as a 
function of the state. This is proved by means of a new 
asymptotic protocol to prepare a state from pre-shared entanglement
and by local operations only.

We go on to employ this convexity property in an investigation of the
additivity of the (single-copy) entanglement of purification: using 
numerical results for two-qubit Werner states we find strong
evidence that the entanglement of purification is different from
its regularization, hence that entanglement of purification is not
additive.
\end{abstract}

\maketitle

\section{Introduction}
\label{sec:intro}
It is well understood that entanglement plays a key role in quantum information
science. The best known applications of quantum entanglement, like superdense 
coding~\cite{dense} and quantum teleportation~\cite{tele}, demonstrate this amply.
The theory of quantum entanglement, which aims at quantifying entanglement, has 
been developed greatly during the past several decades. For a bipartite pure 
state $\psi^{AB} = \proj{\psi}^{AB}$, the von Neumann entropy of the reduced state, 
$S(A) = -\tr \psi^A\log\psi^A$ provides the unique measure of entanglement
where $\psi^A=\tr_B \proj{\psi}^{AB}$.
It is denoted $E(\psi)$, and this number quantifies the asymptotically
faithful conversion rate of many copies of $\psi$
into maximally entangled qubit pairs, and vice versa~\cite{pure-e}.
For mixed state, this asymptotic reversibility is lost, in general the
so-called distillable entanglement is strictly smaller than the 
entanglement cost; see the recent survey~\cite{HHHH-ent} for these
facts and pointers to the vast literature on entanglement quantification.

Motivated by entanglement theory, Terhal \emph{et al.}~\cite{Terhal02}
proposed a measure of total (i.e.~encompassing both quantum and classical) correlations
in a quantum state, called entanglement of purification.

\begin{definition}
  \label{def:EoP}
  Given a bipartite density matrix $\rho^{AB}$ on $A\ox B$, 
  the \emph{entanglement of purification (EoP)} is
  \begin{equation*}\begin{split}
    E_P(\rho) &:= \min E\left(\proj{\psi}^{AA':BB'}\right) \\
              &\phantom{====;} \text{s.t. } \psi^{AA'BB'} \text{ purification of } \rho^{AB},
  \end{split}\end{equation*}
  where $E\left(\proj{\psi}^{AA':BB'}\right) = S(AA')$ is the entanglement
  of the pure state $\psi$ across the bipartite cut $AA':BB'$.
\end{definition}
That the above is really a minimum and not just an infimum follows from 
the fact that w.l.o.g.~the dimensions of $A'$ and $B'$ are bounded in
terms of $|A|$ and $|B|$~\cite{Terhal02}. Indeed, in~\cite{Ibinson08}
it was shown that one may assume
\[
  |A'|, |B'| \leq \operatorname{rank} \rho^{AB} \leq |A||B|.
\]

Entanglement of purification is a genuine measure of total correlation 
in a bipartite state: it is 
non-negative, vanishes precisely on the product states $\rho^{AB}=\rho^A\ox\rho^B$
(which are the only states without any correlations), and is non-increasing
under local operations. Also, it is known to be asymptotically
continuous~\cite{Terhal02}.
Furthermore, it has an operational interpretation as a cost measure.
Namely, it was shown in~\cite{Terhal02} that the entanglement cost of preparing
many copies of a bipartite state $\rho^{AB}$, with the restriction that only a 
vanishing rate of communication is allowed, denoted $E_{\text{LO}_q}(\rho)$, 
equals the regularized entanglement of purification:
\[
  E_{\text{LO}_q}(\rho) = \lim_{n\rightarrow\infty}\frac1n E_P(\rho^{\otimes n})
                        =: E_P^\infty(\rho).
\]
(That a communication $\Theta(\sqrt{n})$ is sufficient and
necessary, even for pure states, was shown by Lo and Popescu~\cite{LoPopescu}
and in~\cite{HarrowLo,HaydenWinter}.)

Hayashi proved that the optimal visible compression rate for mixed states is equal 
to $E_P^\infty$ of a state associated to the ensemble~\cite{Hayashi06}.
More generally, the regularized entanglement of purification characterizes the
communication cost of simulating a channel without prior entanglement~\cite{QRST}
(in contrast to the Quantum Reverse Shannon Theorem).
Furthermore, in~\cite[Theorem 2]{HorodeckiPiani12} entanglement of purification,
or rather its regularization, 
was linked to the maximum advantage a given mixed state yields in dense coding.

However, it is not known how to evaluate the regularized entanglement of purification.
As a matter of fact, it is still an open question whether entanglement of purification is additive, i.e.
\[
  E_P(\rho^{A_1B_1}\otimes \sigma^{A_2B_2}) \stackrel{?}{=} E_P(\rho^{A_1B_1})+E_P(\sigma^{A_2B_2}).
\]
Clearly, a positive answer to this question would imply $E_P^\infty = E_P$,
and thus a single-letter formula for $E_{\text{LO}_q}(\rho)$.
Recently, several similar-looking entanglement quantities and
capacity-like measures were shown to be 
non-additive~\cite{VW01,WernerHolevo02,SmithYard08,HaydenWinter08,Hastings09},
and so one might speculate that the answer to the above question is negative, too.
However, these constructions do not seem to imply anything directly for entanglement of purification.

\begin{remark}
  \label{rem:add}
  $E_P(\psi^{AB})=S(\psi)$ for pure states $\psi=\proj{\psi}$,
  and on product states, $E_P(\rho^A\otimes \rho^B)=0$, so additivity holds 
  for these two classes~\cite{Terhal02}. 
  
  In~\cite{Christandl05} it was shown more generally that $E_P(\rho^{AB})=S(\rho^A)$
  whenever the (pure or mixed) state $\rho$ is supported either on the antisymmetric or the 
  symmetric subspace of $A\ox B$, with $|A|=|B|$.
  So additivity holds for all such states, too.
\end{remark}

In the present paper, we prove results which strongly suggest that entanglement 
of purification may not be additive. 
In Section~\ref{sec:convex-property}, we will introduce a new 
property of the regularized entanglement of purification, which can
be expressed as the convexity of the difference between regularized entanglement of purification
and the entropy of the state, $E_P^\infty(\rho)-S(\rho)$.
Then, in section~\ref{sec:Werner} 
we investigate numerically the functional $E_P(\rho)-S(\rho)$
for the one-parameter family of Werner states on two qubits: since
we find that the latter is not convex, we conclude (except for
gross numerical error) that entanglement of purification is different from its regularization. 
Indeed, our convexity result implies an upper bound on $E_P^\infty$ which
is much smaller than our best estimate for $E_P$ on certain Werner
states.
Finally, in section~\ref{sec:conclusion} we conclude,
highlighting some open questions.

\section{A convexity property of regularized entanglement of purification}
\label{sec:convex-property}
Here we state our main result, a new property of the regularized entanglement
of purification:
\begin{theorem}
  \label{thm:property}
  For a decomposition $\rho^{AB}=\sum_i p_i\rho_i^{AB}$ as an
  ensemble of possibly mixed states $\rho_i$,
  \[
    E_P^{\infty}(\rho^{AB}) \leq \sum_i p_i E_P^{\infty}(\rho_i^{AB}) + \chi(\{p_i;\rho_i\}),
  \]
  where $\chi = \chi(\{p_i;\rho_i\})=S\bigl(\sum_i p_i\rho_i\bigr) - \sum_i p_iS(\rho_i)$ 
  is the Holevo information (cf.~\cite{Wilde:book}).
\end{theorem}
\begin{proof}
We shall describe an asymptotic protocol for creating $\rho^{\otimes n}$, 
using asymptotically optimal ways of generating $\rho_i^{\otimes k_i}$ 
(with $k_i\approx np_i$) as subroutines. In the protocol, the term 
$\sum_i p_i E_P^{\infty}(\rho_i^{AB})$ will be naturally visible as 
the rate of entanglement used, while $\chi(\{p_i;\rho_i\})$ will emerge as 
the rate of classical shared randomness (which of course can be obtained
from entanglement at rate $1$ by measuring).

To be specific, we have
\[
  \rho^{\otimes n} = \sum_{i^n=i_1 i_2\ldots i_n} p_{i^n} \rho_{i^n},
\]
with $p_{i^n} = p_{i_1}p_{i_2}\cdots p_{i_n}$ and 
$\rho_{i^n}=\rho_{i_1}\otimes \rho_{i_2}\otimes \cdots\otimes \rho_{i_n}$.
For a string $i^n = i_1 i_2\ldots i_n$ let $k(i|i^n)$ count the number 
of occurrences of $i$. Then, define the set of typical indices,
\[
  \mathcal{T} := \bigl\{ i^n : \forall i\quad | k(i|i^n) - p_i n | \leq \delta n \bigr\}.
\]
Below we outline the argument to show that there exists a family of
indices, $i^n(1),\ldots,i^n(K) \in \cT$, $K=2^{n(\chi+\delta)}$, such that
\[
  \rho^{\otimes n} \approx \sum_{j=1}^K \frac{1}{K} \rho_{i^n(j)},
\]
the approximation being asymptotically perfect in trace norm.
Then the protocol to create $\rho^{\otimes n}$ goes as follows: The
two parties use $n(\chi+\delta)$ ebits to create the same number of shared
random bits; these are used to sample a uniformly random $i^n(j)$,
$j=1,\ldots,K$. Then for each $i$, they invoke the given protocols to generate
$k_i = k_(i|i^n)$ copies of $\rho_i$, using $k_i(E_P^\infty(\rho_i)+\delta)$ 
ebits and $\text{LO}_q$, thus creating an approximation to $\rho_{i^n(j)}$.
The total entanglement consumption of this protocol is 
\[
  \leq n(\chi+\delta) + n\sum_i p_i E_P^\infty(\rho_i) + n\delta + n\delta\log|A|,
\]
which is what we want, since $\delta>0$ can be made arbitrarily small.

The set $\{i^n(1),\ldots,i^n(K)\}$ is shown to exist by the probabilistic
method: Indeed, we draw the $i^n(j)$ i.i.d.~according to the distribution
$q_{i^n} := \frac{1}{Q}p_{i^n}$ on $\cT$,
with $Q=p^n(\cT)$ the probability of finding a random string $i^n$ in the set $\cT$.
The core part of the proof of the main theorem in~\cite{POVM-compr} 
(Theorem 2, specifically p.~163) shows that this works. The same
technique was used again in~\cite[Proposition~2]{GroismanPopescuWinter:total},
incidentally in a different attempt to quantify total correlations in
a quantum state. Here we give only a summary outline.

We need to introduce some more ``typicality'' notation (cf.~\cite{Wilde:book}
for more details and properties of these notions):
The typical projector $\Pi$ of $\rho^{\otimes n}$ is
\[
  \Pi := \left\{ 2^{-nS(\rho)-\delta' n} \leq \rho^{\otimes n} \leq 2^{-nS(\rho) + \delta' n} \right\},
\] 
the spectral projector corresponding to the typical eigenvalues
of $\rho^{\otimes n}$. 
Finally, the conditional typical projectors $\Pi_{i^n}$ of the states
$\rho_{i^n}=\rho_{i_1}\otimes \rho_{i_2}\otimes \cdots\otimes \rho_{i_n}$:
\[
  \Pi_{i^n} := \left\{ 2^{-n\overline{S}-\delta' n} 
                           \leq \rho_{i^n} \leq 2^{-n\overline{S} + \delta' n} \right\},
\]
where $\overline{S}= \sum_i p_i S(\rho_{i})$.
Consider now the operators
\[
  \rho_{i^n}' := \Pi\, \Pi_{i^n}\rho_{i^n}\Pi_{i^n}\,\Pi,
\]
which have the property that for every $\delta'>0$ one can choose $\delta>0$, 
such that for large enough $n$ and all $i^n\in\cT$,
$\|\rho_{i^n}-\rho_{i^n}'\|_1 \leq o(1)$. Thus,
\[
  \rho^{(n)\prime} := \sum_{i^n\in \mathcal{T}} q_{i^n} \rho_{i^n}'
\]
is supported on the typical subspace and defined such that
$\|\rho^{\otimes n} - \rho^{(n)\prime}\|_1 \leq o(1)$. Define
\[
  \Pi^{\prime} := \left\{ \rho^{(n)\prime} \geq \frac{\epsilon}{\tr\Pi} \right\}
\]
as the spectral projector corresponding to the ``large'' eigenvalues of $\rho^{(n)\prime}$.

Finally let
\begin{displaymath}
  \tilde{\rho} := \Pi^{\prime}\rho^{\prime}\Pi^{\prime}
               =\sum_{i^n\in \mathcal{T}} q_{i^n}\sigma_{i^n},
\end{displaymath}
with $\sigma_{i^n}=\Pi^{\prime} \rho_{i^n}' \Pi^{\prime}$.

Now, observe that, restricted to the support of $\Pi^{\prime}$,
and for $i^n\in \mathcal{T}$,
\begin{align*}
  \tilde{\rho} &\geq 2^{-nS(\rho)-\delta' n}\Pi', \\
  \sigma_{i^n} &\leq 2^{-n\bar{S}+\delta' n}.
\end{align*}
In this situation we can apply the operator sampling lemma in \cite{Ahlswede02} 
and conclude that with high probability, $i^n(1),\ldots,i^n(K) \in \mathcal{T}$
are such  that for large enough $n$, $K \leq 2^{n(\chi+3\delta')}$ and 
\begin{displaymath}
  \left\| \frac{1}{K} \sum\limits_{j=1}^K \sigma_{i^n} - \tilde{\rho}^{(n)} \right\|_1 \leq o(1),
\end{displaymath}
hence similarly the same for the distance of the analogous sum over the
$\rho_{i^n(j)}$, from $\rho^{\otimes n}$. Since $\delta'>0$ was arbitrary,
this concludes the proof.
\end{proof}

\medskip
For our present purposes, we rearrange the terms in the above theorem:
\begin{corollary}
  \label{convex}
  For $\rho^{AB}=\sum\limits_ip_i\rho_i^{AB}$,
  \begin{displaymath}
    E_P^{\infty}(\rho^{AB})-S(\rho^{AB}) 
           \leq \sum_i p_i \bigl( E_P^{\infty}(\rho_i^{AB})-S(\rho_i^{AB}) \bigr).
  \end{displaymath}
  In other words, $E_P^{\infty}(\rho)-S(\rho)$ is a convex function of $\rho$.
  \qed
\end{corollary}

Thus if we can find some examples to show $E_P(\rho)-S(\rho)$ is not convex
on quantum states, then $E_P^{\infty}$ can not be equal to $E_P$, which will 
prove that $E_P$ is not additive on some states.
In the following section, we will present the numerical results for two-qubit 
Werner states (replicating essentially the study of~\cite{Terhal02}), which
indicate that entanglement of purification is not additive.

\section{Two-qubit Werner states}
\label{sec:Werner}
Here we are considering the two-qubit Werner states, arguably
the simplest family of states not covered by the additivity results
mentioned in Remark~\ref{rem:add}:
\[
  W(f) := f\Psi_0 + (1-f)\frac13(\1-\Psi_0),
\]
with the maximally entangled singlet state $\Psi_0$ and $0\leq f\leq 1$
(the singlet fraction).

Already in the original EoP paper~\cite{Terhal02}, the authors performed a
numerical minimization with $|A'|, |B'| \leq 4$, which thanks to~\cite{Ibinson08} 
we know to be sufficient to find $E_P(W(f))$ -- see Fig.~\ref{fig:THLD}.
One way of looking at the minimization that one has to perform is as
follows: Diagonalizing the state, $\rho = \sum_{i=0}^3 \lambda_i \proj{\Psi_i}$,
where $\ket{\Psi_i}$ are the four Bell states, starting with the singlet
$\ket{\Psi_0}$, and $\lambda_0 = \e$, $\lambda_1=\lambda_2=\lambda_3=\frac{1-\e}{3}$.
Then we can write a standard purification
\[
  \ket{\varphi}^{ABA'} = \sum_{i=0}^3 \sqrt{\lambda_i} \ket{\Psi_i}^{AB}\ket{i}^{A'},
\]
and any other purification $\ket{\psi} \in ABA'B'$ of $\rho^{AB}$ we can
obtain as
\begin{equation}
  \label{eq:isometry}
  \ket{\psi}^{ABA'B'} = (\1^{AB}\ox V)\ket{\varphi}^{ABA'},
\end{equation}
with an isometry $V:A' \hookrightarrow A'B'$, described by $64$ complex
numbers (subject to normalization and orthogonality constraints, effectively 
leaving $30+29+27+25 = 121$ independent real parameters). Note that one
could extend the isometry to a unitary $U$ on $A'B'$, 
with $U\ket{\phi}^{A'}\ket{0}^{B'} = V\ket{\phi}$ -- however, this introduces
a large number of spurious variables, in fact more than doubling them to $256$,
which have no impact on the objective function.

\begin{figure}[ht]
\includegraphics*[width=0.52\textwidth]{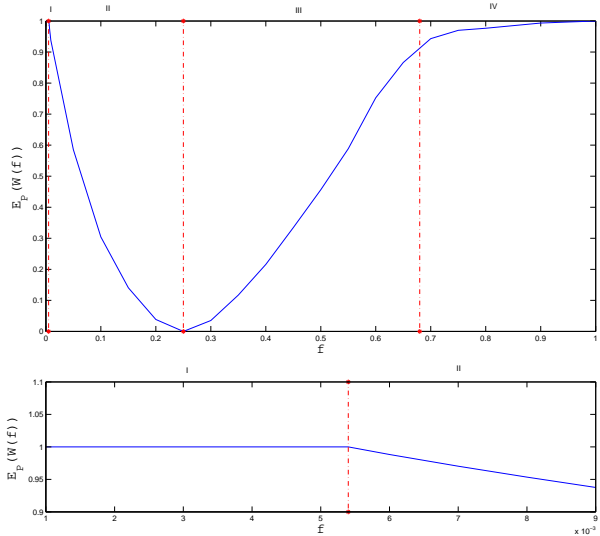}
\caption{The numerical results of~\cite{Terhal02} for $E_P(W(f))$.
    Note that the only
    values known rigorously are at $f=0$ and $f=1$ (both $1$) and
    at $f=\frac14$ ($0$).
    Four different regimes were observed numerically.
    In the first regime, which only extends over a very small range, 
    approximately $0\leq f \leq .005$, the optimal $V$ of eq.~(\ref{eq:isometry}) 
    seems to be the trivial $\ket{\phi}^{A'} \mapsto \ket{\phi}^{A'}\!\ket{0}^{B'}$.
    Thus on this short interval, $E_P(W(f))=1$. 
    In the second regime (roughly $.005 \leq f \leq .25$), entanglement of 
    purification appears convex
    and steeply decreasing with $f$.}
\label{fig:THLD}
\end{figure}

The graph shows an apparent -- concave! -- kink (discontinuity of the first
derivative) at $f \approx .005$. Note that if the kink was real, 
we had achieved our goal, since the entropy $S(W(f))$ is a smooth function on
the open interval $(0,1)$, hence the difference $\Delta(f) = E_P(W(f)) - S(W(f))$
could not possibly be convex as a function of $f$.

Motivated by this observation, we did a re-calculation for $0\leq f \leq .01$.
This revealed that the first regime, where
$E_P(W(f)) \approx 1$, is smaller than it was observed in~\cite{Terhal02};
the range we determined is about $[0, .004]$~\cite{Smolin05}, although the
deviation is tiny. However, we
still see the change from a regime where the function is almost constant $1$
to one where it decreases sharply with $f$. 
In Fig~\ref{fig:concave} we show $\Delta(f)$ and
one can see that indeed it is not convex.

\begin{figure}[ht]
\includegraphics*[width=0.5\textwidth]{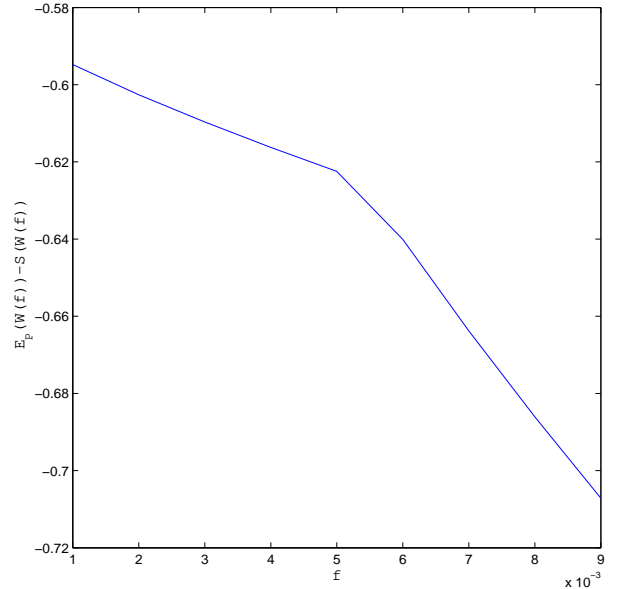}
\caption{$\Delta(f) = E_P(W(f)) - S(W(f))$ for $0 < f < .01$.}
\label{fig:concave}
\end{figure}

We should point out that by using standard minimization algorithms
(local descent with various, usually random, starting points),
we cannot calculate the exact value of entanglement of purification:
What these methods give us are at best \emph{local} minima. 
However, we can treat the local minima from numerics as upper bounds on
the entanglement of purifications, since the algorithm finds 
concrete feasible points with certain values of the objective function
to be minimized:
\[
   E_P(W(0)) = 1 \quad\text{and}\quad E_P(W(.01)) \leq .9226,
\]
showing via theorem~\ref{thm:property} that $E_P^\infty(W(.005)) \leq .9663$.

Put differently, if $E_P(W(.005)) > .9663$, we will have
\[
  \Delta(0.005) > \frac12 \Delta(0) + \frac12 \Delta(.01),
\]
\emph{i.e.}~non-convexity of $E_P(\rho)-S(\rho)$, and thus 
non-additivity of entanglement of purification.

The numerics suggests $E_P(W(.005)) \geq .99$, not even coming close 
to the above value of $.9663$. Hence, unless there is some deep and narrow
``crevasse'' in the landscape of the function $E(\psi^{AA':BB'})$, hiding 
the true minimum value, we are forced to conclude that
$E_P^\infty(W(.005))$ is strictly smaller than $E_P(W(.005))$.

\section{Discussion and conclusions}
\label{sec:conclusion}
Our main new contribution to the study of entanglement of purification,
and its regularization $E_P^\infty = E_{\text{LO}_q}$, is 
theorem~\ref{thm:property}. A special case of its application is when
$\rho^{AB}$ is decomposed into \emph{product states}, meaning that 
$E_P(\rho^{AB}_i) = 0$. Then, the protocol described in the proof 
of the theorem uses only shared randomness, at rate $\chi(\{p_i,\rho_i\})$.
This generalizes a result due to Wyner~\cite{Wyner:CR}
(cf.~\cite{Winter:public-secret} for a more modern account) on the 
creation of a bipartite distribution $P_{XY}$ by local operations
(noisy channels) from limited shared randomness:
\begin{equation*}\begin{split}
  w(P_{XY}) &= \min I(V:XY) \\
            &\phantom{==}
             \text{s.t. } X \text{---} V \text{---} Y \text{ is a Markov chain,}
\end{split}\end{equation*}
where $I(V:XY) = H(XY)-H(XY|V)$ is the Shannon mutual information.

Since the theorem puts a nontrivial bound on the regularized 
entanglement of purification, expressed conveniently as the
convexity of $E_P^\infty-S$, we could use it to probe the additivity
of entanglement of purification. We find that, apart from the possibility of a gross numerical
error, entanglement of purification is non-additive already on certain two-qubit Werner states.
Interestingly, we can only say that for some sufficiently large
$n$, $\frac1n E_P(\rho^{\otimes n}) < E_P(\rho)$, but our proof of
theorem~\ref{thm:property} does not yield directly an estimate for 
this $n$; in any case, we may expect it to be rather large.

The non-additivity of entanglement of purification also answers 
a question from~\cite{HorodeckiPiani12}:
Indeed, our results imply the non-additivity of the ``quantum advantage 
of dense coding'' on some states, via their monogamy
identity~\cite[Theorem 2]{HorodeckiPiani12}.

To come back to our Werner state example: Of course,
it would be most desirable to remove the need for
numerical calculation in the argument. We leave a completely
rigorous proof of the non-additivity of entanglement of purification to 
future work; noting
only that since our example is concrete, and we have a
concrete benchmark,
\[
  E_P(W(.005)) \gtrless .9663,
\]
this could be accomplished in principle by discretization and exhaustive
search over the parameter space. 
The reason we have not done this is that such a brute force approach
is too CPU intensive for practical desktop PC calculations.

In a similar vein, we would like
to find explicit states $\rho$ and $\sigma$ with
\[
  E_P(\rho^{A_1B_1}\otimes \sigma^{A_2B_2}) \neq E_P(\rho^{A_1B_1})+E_P(\sigma^{A_2B_2}).
\]

To end, we remark that our study does not impact on the 
possible non-additivity of $E_P^\infty = E_{\text{LO}_q}$, which we
recommend to the reader as an interesting problem in itself. Even
more interesting however is the problem of finding a tractable
(or even ``single-letter'') expression for $E_P^\infty$, which
in a certain sense would generalized Wyner's beautiful answer for the
classical randomness cost of probability distributions 
$P_{XY}$~\cite{Wyner:CR,Winter:public-secret}.

\acknowledgments
The authors thank Fernando Brand\~{a}o and Jonathan Oppenheim for conversations
on the entanglement of purification.

The work of JC is supported by NSERC and NSF of China (Grant No. 61179030).
AW is supported by the European Commission (STREP ``QCS'' and
Integrated Project ``QESSENCE''), the ERC (Advanced Grant ``IRQUAT''),
a Royal Society Wolfson Merit Award and a Philip Leverhulme Prize.
The Centre for Quantum Technologies is funded by the Singapore
Ministry of Education and the National Research Foundation as part
of the Research Centres of Excellence programme.

\bibliographystyle{unsrt}

\end{document}